%% file: main.tex
\documentclass[runningheads]{lmcs}
\pdfoutput=1

\usepackage{lastpage}
\lmcsdoi{17}{4}{12}
\lmcsheading{}{\pageref{LastPage}}{}{}%
{Jan.~01,~2021}{Nov.~24,~2021}{}

\usepackage{hyperref}
\usepackage{cleveref}


\usepackage{graphicx}
\usepackage[utf8]{inputenc}

\usepackage{amsmath}
\usepackage{amssymb}
\usepackage{mathrsfs}
\usepackage{wrapfig}
\usepackage{mathcommand}
\usepackage{stmaryrd}
\usepackage{cleveref}
\usepackage{xcolor}
\usepackage{marvosym}
\usepackage{xspace}
\usepackage{tikz}
\usetikzlibrary{calc,decorations.pathmorphing,shapes}

\hypersetup{hidelinks}
\hypersetup{breaklinks}
\usepackage[notion,quotation, paper]{knowledge}
\usepackage{mathcommand}

\input{knowledge-flow.tex}
\input{macros.tex}

\begin{document}

\title{Controlling a random population}
\titlecomment{{\lsuper*}A conference version was published in FoSSaCS'2020~\cite{ColcombetFijalkowOhlmann20}.
This work was supported by the European Research Council (ERC) under the European Union's Horizon 2020 research and innovation programme (grant agreement No.670624), and by the DeLTA ANR project (ANR-16-CE40-0007).}

\author[T.~Colcombet]{Thomas Colcombet\rsuper{a}}
\address{CNRS, Universit{\'e} de Paris, France}

\author[N.~Fijalkow]{Nathana{\"e}l Fijalkow\rsuper{b}}
\address{CNRS, LaBRI, Universit{\'e} de Bordeaux, France \\
and the Alan Turing Institute of data science, London, United Kingdom}

\author[P.~Ohlmann]{Pierre Ohlmann\rsuper{c}}
\address{Universit{\'e} de Paris, France}

\maketitle             

\begin{abstract}
Bertrand et al.\ introduced a model of parameterised systems, where each agent is represented by a finite state system,
and studied the following control problem: 
for any number of agents, does there exist a controller able to bring all agents to a target state?
They showed that the problem is decidable and \textsc{EXPTIME}-complete in the adversarial setting,
and posed as an open problem the stochastic setting,
where the agent is represented by a Markov decision process.
In this paper, we show that the stochastic control problem is decidable.
Our solution makes significant uses of well quasi orders, of the max-flow min-cut theorem, and of the theory of regular cost functions.
We introduce an intermediate problem of independence interest called the sequential flow problem and study its complexity.
\end{abstract}

\section{Introduction}\label{sec:introduction}

\input{introduction.tex}

\section{The stochastic control problem}\label{sec:problem}

\input{problem}

\section{The sequential flow problem}\label{sec:sequential_flow}

\input{sequential_flow}

\section{Reduction from the stochastic control problem \texorpdfstring{\\}{}to the sequential flow problem}\label{sec:reduction}

\input{reduction}

\section{Solution to the simple sequential flow problem}\label{sec:decidability_simple}

\input{decidability_simple}

\section{Solution to the sequential flow problem}\label{sec:decidability}

\input{decidability}

\section{Lower bound for the simple sequential flow problem}\label{sec:complexity}

\input{complexity}

\section{Conclusions}\label{sec:conclusions}

\input{conclusions}

\section*{Acknowledgements}
We thank Nathalie Bertrand and Blaise Genest for introducing us to this fascinating problem, 
and the preliminary discussions at the Simons Institute for the Theory of Computing in Fall 2015. 
We also thank Arnaud Sangnier and Mahsa Shirmohammadi for interesting discussions about the sequential flow problem.

\bibliographystyle{alpha}
\bibliography{bib.bib}

\end{document}

%% file: knowledge-flow.tex
\knowledge{sink}[Sink|sink state]{}

\knowledge{notion}
 | $(s,t)$-flow
 | $((s,0),(t,\ell))$-flow
 | $((s,0),t)$-flow
 
\knowledge{notion}
 | desert automaton
 | desert automata

\knowledge{notion}
 | value
 
\knowledge{notion}
 | decomposition
 | decompositions
 
\knowledge{notion}
 | $(s,t)$-cut
 | $((s,0),(t,\ell))$-cut
 | $((s,0),t)$-cut
 
\knowledge{notion}
 | conservation law
 | conservation laws

\knowledge{notion}
 | almost sure reachability problem
\knowledge{notion}
 | Sequential Flow Problem
 | Sequential Flow
 | sequential flow
 | sequential flow problem
 | Sequential flow problem
\knowledge{notion}
 | Simple sequential flow problem
 | simple sequential flow problem
 | simple SFP
\knowledge{notion}
 | NFA intersection emptiness
 | NFA intersection
\knowledge{notion}
 | boundedness problem
 | boundedness question
\knowledge{notion}
 | stochastic control problem
 | Stochastic control problem

\knowledge{notion}
 | theory of regular cost functions
 | regular cost function
 | Regular cost functions
 | regular cost functions

\knowledge{strategy}{notion}
\knowledge{transition relation}{notion}
\knowledge{almost surely}{notion}

\knowledge {downward closed subset}{notion}
\knowledge {downward closed set}{synonym}
\knowledge {downward closed subsets}{synonym}
\knowledge {downward closed sets}{synonym}
\knowledge {Downward closed sets}{synonym}
\knowledge {downward closed}{synonym}

\knowledge {decomposition into ideals}[decomposed]{notion}
\knowledge {decomposition}{synonym}

\knowledge {ideal}[ideals|Ideals]{notion}

\knowledge {complexity}{notion}

\knowledge {place}[places|Places]{notion}

\knowledge {configuration}[configurations|Configurations]{notion}

\knowledge {flow}[Flow|flows|Flows]{notion}
\knowledge {capacity}[Capacity|capacities|Capacities]{notion}

\knowledge {flow word}[flow words|Flow words|word of flows]{notion}

\knowledge {capacity word}[capacity words|Capacity words]{notion}
\knowledge {word of capacities}{synonym}

\knowledge {source configuration}{notion}

\knowledge {target configuration}{notion}

\knowledge {flow equations}{notion}

\knowledge {flow word from~$\c $ to~$\c '$}{notion}

\knowledge {run of the system}{notion}

\knowledge {irreducible}{notion}

\knowledge {cost function}[cost functions |functions]{notion}

\knowledge {max-flow min-cut theorem}[max-flow min-cut Theorem | maximal flow]{notion}

\knowledge {Markov decision process}[Markov decision processes]{notion}  
\knowledge {MDP}[MDPs]{notion}

\knowledge{goes to}[goes from|going from]{notion}

\knowledge{gadget}{notion}

\knowledge{parameterised control}{notion}

\knowledge{transition table}[boolean transition table |transition |probabilities|transitions]{notion}

\knowledge{state}[states | States]{notion}
  
\knowledge{action}[actions | Actions]{notion}

\knowledge{source}[source state]{notion}
\knowledge{target}[target state]{notion}
\knowledge{tokens}[token | Token | Tokens]{notion}

\knowledge{configuration}[Configuration | configurations | Configurations]{notion}
  
\knowledge{interruption}[interruptions | interrupt | interrupts |countering]{notion}

\knowledge{order}[partial order]{notion}

\knowledge{well quasi order}[well quasi orders | well quasi ordered set]{notion}

\knowledge{WQO}[WQOs]{link=well quasi order}

\knowledge{cut}[Cuts | cuts]{notion}
  
\knowledge{cost of a cut}[cost|weight]{notion}

\knowledge{agree}[Agree | agrees]{notion}

\knowledge{Controller}[controller]{notion}

\knowledge{notion}
 | distance automaton
 | distance automata
 | Distance automata
 | automaton

\knowledge{notion}
 | Monadic second order logic
 | monadic second order logic
 | MSO

\knowledge{notion}
 | Cost monadic second order logic
 | cost monadic second order logic
 | cost monadic logic
 | Cost monadic logic
 | cost MSO
 | Cost MSO

\declaremathcommandPIE\BuchiG{\kl[\BuchiG#1#3]{\mathcal G}_{\!\mdp}#1\IfEmptyTF{#3}{}{^{(\GetExponent{#3})}}}
\knowledge{\BuchiG}{notion}
\knowledge{\BuchiG^{0}}[\BuchiG^{i} | \BuchiG^{i+1} | \BuchiG^{i-1}]{notion}

\declaremathcommand\targetstate{\kl[\reductiontargetstate]{t}}
\knowledge\reductiontargetstate{notion}

%% file: macros.tex
\newcommand{\val}{\text{val}}
\newcommand{\sem}[1]{\llbracket #1 \rrbracket}

\theoremstyle{plain}\newtheorem{theorem}[thm]{Theorem}
\theoremstyle{plain}\newtheorem{lemma}[thm]{Lemma}
\theoremstyle{plain}
\theoremstyle{plain}\newtheorem{problem}[thm]{Problem}
\theoremstyle{definition}\newtheorem{definition}[thm]{Definition}
\knowledgeconfigureenvironment{theorem,lemma,proof,proposition,fact,problem,definition}{}

\newmathcommandPIE\Nats{\mathbb{N}#1#2#3}

\LoopCommands{abcdefghijklmnopqrstuvwxyzABCDEFGHIJKLMNOPQRASTUVWXYZ}[#1][li#1][rli#1][pli#1]
  {\declaremathcommandPIE#2{#1##1##2##3}
   \declaremathcommandPIE#3{#1##1##2##3}
   \declaremathcommandPIE#4{#1##1##2##3}
   \declarecommandPIE#5{}
  }

\LoopCommands{ABCDEFGHIJKLMNOPQRASTUVWXYZ}[c#1][lic#1][plic#1]
  {\declaremathcommandPIE#2
       {\kl[{#2##1##2}]{\mathcal{#1}##1##2##3}}
   \declaremathcommandPIE#3
      {\intro[{#2##1##2}]{\mathcal{#1}##1##2##3}%
      \knowledge{#2##1##2}{notion}}
   \declarecommandPIE#4
      {\phantomintro[{#2##1##2}]{\mathcal{#1}##1##2##3}%
      \knowledge{#2##1##2}{notion}}
  }

\knowledgedirective{math notion}{notion}

\newmathcommand\complexity[1]{\kl[\complexity]{||}#1\kl[\complexity]{||}}
\knowledge\complexity{math notion}

\newmathcommandPIE\Flows{\kl[\Flows#3]{\mathit{Flows}#3}#1#2}
\knowledge\Flows{math notion}
\knowledge{\Flows^{0}}{math notion}
\knowledge{\Flows^{i}}{math notion}
\newmathcommand\decFlows{\kl[\decFlows]{\mathit{Flows}}}
\knowledge\decFlows{math notion}
\newmathcommand\simpleFlows{\kl[\simpleFlows]{\mathit{Flows}}}
\knowledge\simpleFlows{math notion}

\newmathcommand\Capacities{\kl[\Capacities]{\mathit{Capacities}}}
\knowledge\Capacities{math notion}
\newmathcommand\Confs{\kl[\Confs]{\mathit{Confs}}}
\knowledge\Confs{math notion}

\newmathcommand\eFlows{\kl[\eFlows]{\mathit{Flows_{\omega}}}}
\knowledge\eFlows{math notion}
\newmathcommand\eConfs{\kl[\eConfs]{\mathit{Confs_{\omega}}}}
\knowledge\eConfs{math notion}

\newrobustcmd\flowarrow[1]{\mathbin{\kl[\flowarrow]{\rightarrow}^{#1}}}
\knowledge\flowarrow{math notion}
\newrobustcmd\flowpath[1]{\mathbin{\kl[\flowpath]{\leadsto}^{#1}}}
\knowledge\flowpath{math notion}
\newrobustcmd\flow{\kl[{\flow}]{\mathsf{flow}}}
\knowledge\flow{math notion}
\newrobustcmd\weight{\kl[{\weight}]{\mathsf{weight}}}
\knowledge\weight{math notion}

\newrobustcmd\Graph{\kl[\Graph]{\mathsf{Graph}}}
\knowledge\Graph{math notion}

\newmathcommand\source{\kl[\source]{\mathsf{source}}}
\newmathcommand\target{\kl[\target]{\mathsf{target}}}
\newmathcommand\dclosure[1]{\kl[\dclosure]{D}(#1)}

\knowledge\source{math notion}
\knowledge\target{math notion}

\newcounter{sarrow}

\newcommand{\st}{\mathcal Q}
\newcommand{\act}{\mathcal A}
\newcommand{\distr}[1]{\mathcal D \left( #1 \right)}
\newcommand{\pow}[1]{\mathcal P \left( #1 \right)}
\newcommand{\mdp}{\mathcal M}

\newrobustcmd{\dc}[1]{#1\kl[\dc]{\downarrow}}
\knowledge\dc{math notion}

\newmathcommandPIE\Win{X#1#2\IfEmptyTF{#3}{}{^{(\GetExponent{#3})}}}

\newcommand{\Eve}{\kl[\Eve]{Eve}\xspace}
\knowledge\Eve{notion}
\newcommand{\Adam}{\kl[\Adam]{Adam}\xspace}
\knowledge\Adam{notion}

\newcommand{\fin}{F}

\newcommand{\finite}{\mathrm{fin}}
\newcommand{\Jc}{\mathcal J}

\DeclareMathOperator\supp{\kl[\supp]{supp}}
\knowledge\supp{notion}
\DeclareMathOperator\pre{\kl[\pre]{pre}}
\knowledge\pre{notion}
\DeclareMathOperator\post{\kl[\post]{post}}
\knowledge\post{notion}
\DeclareMathOperator\IsADecomp{\mathtt{\kl[\IsADecomp]{IsDecomp}}}
\knowledge\IsADecomp{notion}

\newcommand{\gadg}[1]{\mathcal S_{ #1 }}

\newcommand{\lef}{\mathrm{left}}
\newcommand{\righ}{\mathrm{right}}

\newcommand{\sr}{\mathrm{init}}
\newcommand{\fr}{\mathrm{fin}}

\DeclareMathOperator\Disc{\kl[\Disc]{\mathtt{Disc}}}
\knowledge\Disc{notion}

\newcommand{\Capa}{\mathscr C}

\newcommand{\conf}[1]{#1}
\newcommand{\fl}[1]{#1}
\newcommand{\bfl}[1]{\mathbf{#1}}

\newcommand{\set}[1]{\left\{ #1 \right\}}

\newmathcommand\SFP{\kl[\SFP]{\mathrm{Pre}^*}}
\knowledge{\SFP}{notion}

\newcommand{\aut}{\mathcal A}

\declaremathcommand\redmdp{\kl[\reductionmdp]{\mathcal M}}
\knowledge{\reductionmdp}{notion}

\declaremathcommand\actions{\mathcal A}
\declaremathcommand\states{\mathcal Q}
\declaremathcommand\transitions{\Delta}

\declaremathcommand\cfeq{\mathrel{\kl[\cfeq]{\approx}}}
\knowledge\cfeq{math notion}

\newcommand\afin{\kl[\afin]{a_{\finp}}}
\knowledge\afin{math notion}
\newcommand\finp{\mathrm{fin}} 


%% file: introduction.tex

\paragraph*{The control problem for populations of identical agents.}
The model we study was introduced in~\cite{BDGG17} (see also the journal version~\cite{BDGGG19}):
a population of agents are controlled uniformly, meaning that the controller applies the same action to every agent.
The agents are represented by a finite state system, the same for every agent.
The key difficulty is that there is an arbitrary large number of agents:
the control problem is whether for every $n \in \mathbb{N}$, there exists a controller able to bring all $n$ agents synchronously to a target state.

The technical contribution of~\cite{BDGG17,BDGGG19} is to prove that 
in the adversarial setting where an opponent chooses the evolution of the agents, the (adversarial) control problem is EXPTIME-complete.

In this paper, we study the stochastic setting, where each agent evolves independently according to a probabilistic distribution,
\textit{i.e.} the finite state system modelling an agent is a Markov decision process.
The control problem becomes whether for every $n \in \mathbb{N}$, there exists a controller able 
to bring all $n$ agents synchronously to a target state with probability one.

\medskip

Our main technical result is that the stochastic control problem is decidable.
In the next paragraphs we discuss four motivations for studying this problem:
control of biological systems, parameterised verification and control, distributed computing, and automata theory.

\paragraph*{Modelling biological systems.}
The original motivation for studying this model was for controlling population of yeasts (\cite{UMDCFBHB15}).
In this application, the concentration of some molecule is monitored through fluorescence level. 
Controlling the frequency and duration of injections of a sorbitol solution influences the concentration of the target molecule,
triggering different chemical reactions which can be modelled by a finite state system.
The objective is to control the population to reach a predetermined fluorescence state.
As discussed in the conclusions of~\cite{BDGG17,BDGGG19}, the stochastic semantics 
is more satisfactory than the adversarial one for representing the behaviours of chemical reactions,
so our decidability result is a step towards a better understanding of the modelling of biological systems
as populations of arbitrarily many agents represented by finite state systems.

\paragraph*{From parameterised verification to ""parameterised control"".}
Parameterised verification was introduced in~\cite{GS92}: 
the goal is to prove the correctness of a system specification regardless of the number of its components.
We refer to~\cite{AbdullaD16} for a recent survey on the topic.
The control problem we study here, which was introduced in~\cite{BDGG17,BDGGG19}, is a step towards \textit{parameterised control}:
the goal is to control a system composed of many identical components in order to ensure a given property.

\paragraph*{Distributed computing.}
Our model resembles two models introduced for the study of distributed computing.
The first and most widely studied is population protocols, introduced in~\cite{AADFP06}:
the agents are modelled by finite state systems and interact by pairs drawn at random.
The mode of interaction is the key difference with the model we study here: in a time step,
all of our agents perform simultaneously and independently the same action.
This brings us closer to broadcast protocols as studied for instance in~\cite{EFM99}, 
in which one action involves an arbitrary number of agents.
As explained in~\cite{BDGG17,BDGGG19}, our model can be seen as a subclass of (stochastic) broadcast protocols,
but key differences exist in the semantics, making the two bodies of work technically independent.

The focus of the distributed computing community when studying population or broadcast protocols is to construct the most efficient protocols
for a given task, such as (prominently) electing a leader.
A growing literature from the verification community focuses on checking the correctness of a given protocol against a given specification;
we refer to the recent survey~\cite{Esparza16} for an overview.
We concentrate on the control problem, which can then be seen as a first result in the control of distributed systems in a stochastic setting.

\paragraph*{Alternative semantics for probabilistic automata.}
It is very tempting to consider the limit case of infinitely many agents: 
the parameterised control question becomes the value $1$ problem for probabilistic automata, 
which was proved undecidable in~\cite{GO10}, and even in very restricted cases, see for instance~\cite{FGHO14}.
Hence abstracting continuous distributions by a discrete population of arbitrary size 
can be seen as an approximation technique for probabilistic automata.
Using $n$ agents corresponds to using numerical approximation up to $2^{-n}$ with random rounding;
in this sense the control problem considers a converging approximation scheme.
The plague of undecidability results on probabilistic automata (see \textit{e.g.} \cite{Fijalkow17})
is nicely contrasted by our positive result, 
which is one of the few decidability results on probabilistic automata not making structural assumptions on the underlying graph.

\paragraph*{Our results.}
We prove decidability of the stochastic control problem.
The first insight is given by the theory of well quasi orders, 
which motivates the introduction of a new problem called the sequential flow problem.
The first step of our solution is to reduce the stochastic control problem to (many instances of) the sequential flow problem.
The second insight comes from the theory of regular cost functions, 
providing us with a set of tools for addressing the key difficulty of the problem, namely the fact that there are arbitrarily many agents.
Our key technical contribution is to show the computability of the sequential flow problem by reducing it to a boundedness question
expressed with distance and desert automata using the max-flow min-cut theorem.

\paragraph*{Studying infinite-state Markov decision processes.}
The technical contribution of this paper is an algorithm for solving an infinite-state Markov decision process with a finite representation. The purpose of the notion of decisive Markov chains~\cite{AHM07} is to define a unifying property for studying infinite-state Markov chains with ``finite-like'' properties. The Markov decision processes we consider in this paper do not yield decisive Markov chains, which suggests that our results are in some sense orthogonal to this notion.

\paragraph*{Organisation of the paper.}
We define the "stochastic control problem" in Section~\ref{sec:problem}, and the "sequential flow problem" in Section~\ref{sec:sequential_flow},
together with its simple variant.
We construct a reduction from the former to (many instances of) the latter in Section~\ref{sec:reduction}.
Our solution of the "sequential flow problem" (SFP) goes in two steps, first studying the simple variant in Section~\ref{sec:decidability_simple},
and then the general case in Section~\ref{sec:decidability}. 
Section~\ref{sec:complexity} proves the \textsc{PSPACE}-hardness of the "simple sequential flow problem".

%% file: problem.tex


For a finite set $X$, we let $\distr{X}$ denote the set of probabilistic distributions over $X$,
\textit{i.e.} functions $\delta : X \to [0,1]$ such that $\sum_{x \in X} \delta(x) = 1$.

\begin{definition}\label{def:mdp}
\AP A ""Markov decision process"" (""MDP"" for short) consists of
\begin{itemize}
\item a finite set of ""states"" $\st$,
\item a finite set of ""actions"" $\act$,
\item a stochastic ""transition table"" $\rho : \st \times \act \to \distr{\st}$.
\end{itemize}
\end{definition}
\AP The interpretation of the "transition table" is that from the state $p$ under action $a$, 
the probability to transition to $q$ is $\rho(p,a)(q)$.
The ""transition relation"" $\Delta$ is defined by
\[
\Delta = \set{ (p,a,q) \in \st \times \act \times \st : \rho(p,a)(q) > 0 }.
\] 
We also use $\Delta_a$ given by 
$\{(p,q) \in \Q \times \Q : (p,a,q) \in \Delta\}$.

\AP We refer to~\cite{Kucera11} for the usual notions related to "MDPs"; 
it turns out that very little probability theory will be needed in this paper, 
so we restrict ourselves to mentioning only the relevant objects.
In an "MDP" $\mdp$, a strategy is a function $\sigma : \st \to \act$; note that we consider only pure and positional strategies,
as they will be sufficient for our purposes (see Section~\ref{sec:reduction}).
\AP Given a ""source"" $s \in \st$ and a ""target"" $t \in \st$, 
we say that the "strategy" $\sigma$ ""almost surely"" reaches $t$ if 
the probability that a path starting from $s$ and consistent with $\sigma$ eventually leads to $t$ is 1.
As we shall recall in Section~\ref{sec:reduction}, 
whether there exists a strategy ensuring to reach $t$ almost surely from $s$, called the 
""almost sure reachability problem"" for "MDP" can be reduced to solving a non-stochastic two player B{\"u}chi game,
and in particular does not depend upon the exact "probabilities".
In other words, the only relevant information for each $(p,a,q) \in \st \times \act \times \st$ is whether $\rho(p,a)(q) > 0$ or not, equivalently $(p,a,q) \in \Delta$.
Since the same will be true for the stochastic control problem we study in this paper, 
in our examples we do not specify the exact "probabilities", and an edge from $p$ to $q$ labelled $a$ means that 
$\rho(p,a)(q) > 0$.
\AP When considering a target we implicitly assume that it is a sink: it has a single outgoing transition which is a self-loop.

\AP Let us now fix an "MDP" $\mdp$ and consider a population of $n$ ""tokens"" (we use "tokens" to represent the agents).
Each token evolves in an independent copy of the "MDP" $\mdp$.
The controller acts through a ""strategy"" $\sigma : \st^n \to \act$, meaning that given the state each of the $n$ tokens is in,
the controller chooses \textit{one} action to be performed by all tokens independently.
Formally, we are considering the product "MDP" $\mdp^n$ whose set of states is $\st^n$, set of actions is $\act$, and 
"transition table" is $\rho^n(u,a)(v) = \prod_{i = 1}^n \rho(u_i,a)(v_i)$, 
where $u,v \in \st^n$ and $u_i,v_i$ are the $i$\textsuperscript{th} components of~$u$ and $v$.

Let $s,t \in \st$ be the "source" and "target" states, we write $n \cdot s$ and $n \cdot t$ for the constant $n$-tuples where all components are $s$ and $t$.
This ``additive notation'' will be more convenient later on than the ``product notation'' suggested by $\mdp^n$.
For a fixed value of $n$, whether there exists a strategy ensuring to reach $n \cdot t$ almost surely from $n \cdot s$ can be reduced to solving a two player B{\"u}chi game by considering the "MDP" $\mdp^n$ and reducing it to a B{\"u}chi game as explained above.
The stochastic control problem asks whether this is true for arbitrary values of $n$:

\begin{problem}[""Stochastic control problem""]\label{problem:main}\hfill

\textbf{Input}: an "MDP" $\mdp$, a "source" state $s \in \st$, and a "target" state $t \in \st$.

\AP 
\textbf{Output}: Yes if for all $n \in \mathbb{N}$, there exists a "strategy" ensuring to reach $n \cdot t$ almost surely from~$n \cdot s$, and No otherwise.
\end{problem}

Our main result is the following.

\begin{theorem}\label{thm:main}
The "stochastic control problem" is decidable.
\end{theorem}

The fact that the problem is co-recursively enumerable is easy to see: if the answer is ``no'',
there exists $n \in \mathbb{N}$ such that there exist no strategy ensuring to reach $n \cdot t$ almost surely from $n \cdot s$.
Enumerating the values of $n$ and solving the "almost sure reachability problem" for $\mdp^n$ eventually finds this out.
However, it is not clear whether one can exhibit an upper bound on such a witness $n$, which would yield a simple (yet inefficient!) algorithm.
As a corollary of our analysis we can indeed derive such an upper bound, but it is non elementary in the size of the "MDP".

In the remainder of this section we present a few interesting examples.

\begin{exa}\label{ex:one}
Let us consider the "MDP" represented in Figure~\ref{fig:ex_one}.
We show that for this "MDP", for any $n \in \mathbb{N}$, the controller has a strategy almost surely reaching $n \cdot t$ from $n \cdot s$.
Starting with $n$ "tokens" on $s$, we iterate the following strategy:
\begin{itemize}
\item Repeatedly play "action" $a$ until all "tokens" are in $q$;
\item Play "action" $b$.
\end{itemize}
The first step is eventually successful with probability one, since at each iteration there is a positive probability that the number of "tokens" in state $q$ increases.
In the second step, with non zero probability at least one "token" goes to $t$, while the rest goes back to $s$.
It follows that each iteration of this strategy increases with non zero probability the number of "tokens" in $t$. 
Hence, all "tokens" are eventually transferred to $n \cdot t$ almost surely.

\begin{figure}[ht]
    \centering
    \def\svgwidth{\columnwidth}
    \resizebox{0.45\textwidth}{!}{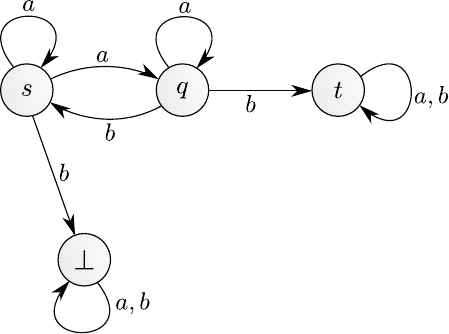}
\caption{The controller can almost surely reach $n \cdot t$ from $n \cdot s$, for any $n \in \mathbb{N}$.}
\label{fig:ex_one}
\end{figure}
\end{exa}

\begin{exa}\label{ex:two}

We now consider the "MDP" represented in Figure~\ref{fig:ex_two}. 
By convention, if from a state some action does not have any outgoing transition (for instance the action $u$ from $s$),
then it goes to the "sink" state $\bot$.

We show that there exists a controller ensuring to transfer seven "tokens" from $s$ to $t$,
but that the same does not hold for eight "tokens". 
For the first assertion, we present the following strategy:
\begin{itemize}
\item Play $a$. One of the "states" $q_1^{i_1}$ for $i_1 \in \{u,d\}$ receives at least $4$ "tokens".
\item Play $i_1 \in \{u,d\}$. At least 4 "tokens" go to $t$ while at most 3 go to $q_1$.
\item Play $a$. One of the "states" $q_2^{i_2}$ for $i_2 \in \{u,d\}$ receives at least $2$ "tokens".
\item Play $i_2 \in \{u,d\}$. At least 2 "tokens" go to $t$ while at most 1 "token" goes to $q_2$.
\item Play $a$. The token (if any) goes to $q_3^i$ for $i_3 \in \{u,d\}$.
\item Play $i_3 \in \{u,d\}$. The remaining token (if any) goes to $t$.
\end{itemize}
Now assume that there are 8 "tokens" or more on $s$. 
The only choices for a strategy are to play $u$ or $d$ on the second, fourth, and sixth moves.
First, with non zero probability at least 4 "tokens" are in each of $q_1^i$ for $i \in \{u,d\}$. 
Then, whatever the choice of "action" $i \in \{u,d\}$, there are at least 4 "tokens" in $q_1$ after the next step. 
Proceeding likewise, there are at least 2 "tokens" in $q_2$ with non zero probability two steps later. 
Then again two steps later, at least 1 "token" falls in the "sink" with non zero probability.

\begin{figure}[ht]
    \centering
    \def\svgwidth{\columnwidth}
    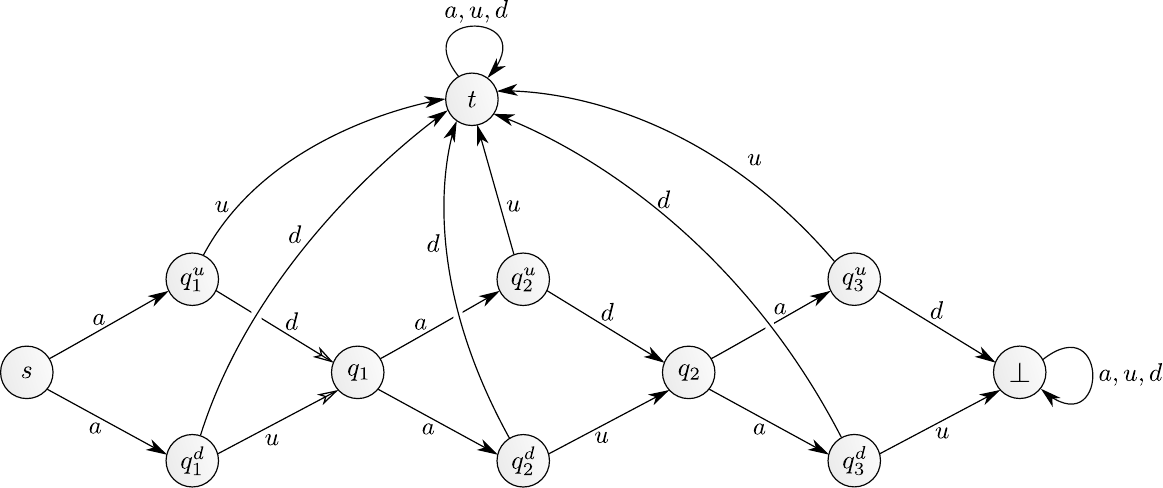
\caption{The controller can synchronise up to $7$ tokens on the target state $t$ almost surely, but not more.}
\label{fig:ex_two}
\end{figure}
Generalising this example shows that if the answer to the "stochastic control problem" is ``no'', 
the smallest number of "tokens" $n$ for which there exist no almost surely strategy for reaching $n \cdot t$ from $n \cdot s$
may be exponential in $|\st|$.
This can be further extended to show a doubly exponential in $\st$ lower bound, as done in~\cite{BDGG17,BDGGG19};
the example produced there holds for both the adversarial and the stochastic setting. 
Interestingly, for the adversarial setting this doubly exponential lower bound is tight.
Our proof for the stochastic setting yields a non-elementary upper bound, leaving a very large gap.
\end{exa}

\begin{exa}\label{ex:three}
We consider the "MDP" represented in Figure~\ref{fig:ex_three}. 
For any $n \in \mathbb{N}$, there exists a strategy almost surely reaching $n \cdot t$ from $n \cdot s$.
We iterate the following strategy:
\begin{itemize}
\item Repeatedly play "action" $a$ until exactly $1$ "token" is in $q_1$.
\item Play "action" $b$. The "token" goes to $q_i$ for some $i \in \{l,r\}$.
\item Play "action" $i \in \{l,r\}$, which moves the "token" to $t$.
\end{itemize}
Note that the first step may take a very long time (the expectation of the number of $a$s to be played until this happens is exponential in the number of "tokens"), but it is eventually successful with probability one.
This very slow strategy is necessary: if $q_1$ contains at least two "tokens", then action $b$ should not be played: 
with non zero probability, at least one "token" ends up in each of $q_l,q_r$, so at the next step some "token" ends up in $\bot$. 
It follows that any strategy almost surely reaching $n \cdot t$ can play $b$ only when there is at most one token in $q_1$. 
This is a key example for understanding the difficulty of the stochastic control problem.

\begin{figure}[ht]
    \centering
    \def\svgwidth{\columnwidth}
    \resizebox{0.45\textwidth}{!}{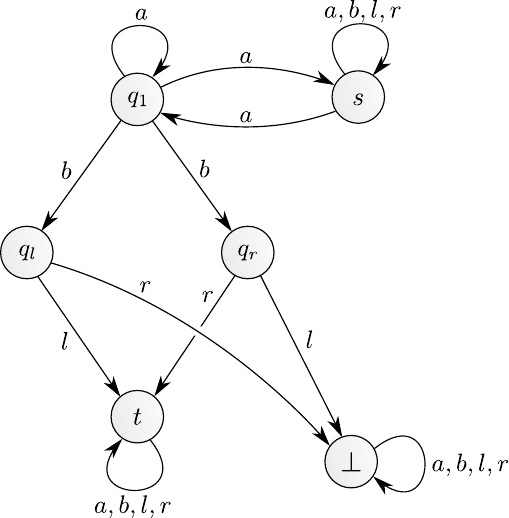}
\caption{The controller can synchronise any number of tokens almost surely on the target state $t$, but since anytime at most one token can be in $\set{q_l, q_r}$, each token must be sent alone from $q_1$ using $b$.}
\label{fig:ex_three}
\end{figure}
\end{exa}

%% file: mdp1.pdf_tex
\begingroup%
  \makeatletter%
  \providecommand\color[2][]{%
    \errmessage{(Inkscape) Color is used for the text in Inkscape, but the package 'color.sty' is not loaded}%
    \renewcommand\color[2][]{}%
  }%
  \providecommand\transparent[1]{%
    \errmessage{(Inkscape) Transparency is used (non-zero) for the text in Inkscape, but the package 'transparent.sty' is not loaded}%
    \renewcommand\transparent[1]{}%
  }%
  \providecommand\rotatebox[2]{#2}%
  \newcommand*\fsize{\dimexpr\f@size pt\relax}%
  \newcommand*\lineheight[1]{\fontsize{\fsize}{#1\fsize}\selectfont}%
  \ifx\svgwidth\undefined%
    \setlength{\unitlength}{215.07101945bp}%
    \ifx\svgscale\undefined%
      \relax%
    \else%
      \setlength{\unitlength}{\unitlength * \real{\svgscale}}%
    \fi%
  \else%
    \setlength{\unitlength}{\svgwidth}%
  \fi%
  \global\let\svgwidth\undefined%
  \global\let\svgscale\undefined%
  \makeatother%
  \begin{picture}(1,0.74374898)%
    \lineheight{1}%
    \setlength\tabcolsep{0pt}%
    \put(0,0){\includegraphics[width=\unitlength,page=1]{mdp1.pdf}}%
  \end{picture}%
\endgroup%

%% file: mdp2.pdf_tex
\begingroup%
  \makeatletter%
  \providecommand\color[2][]{%
    \errmessage{(Inkscape) Color is used for the text in Inkscape, but the package 'color.sty' is not loaded}%
    \renewcommand\color[2][]{}%
  }%
  \providecommand\transparent[1]{%
    \errmessage{(Inkscape) Transparency is used (non-zero) for the text in Inkscape, but the package 'transparent.sty' is not loaded}%
    \renewcommand\transparent[1]{}%
  }%
  \providecommand\rotatebox[2]{#2}%
  \newcommand*\fsize{\dimexpr\f@size pt\relax}%
  \newcommand*\lineheight[1]{\fontsize{\fsize}{#1\fsize}\selectfont}%
  \ifx\svgwidth\undefined%
    \setlength{\unitlength}{557.77064526bp}%
    \ifx\svgscale\undefined%
      \relax%
    \else%
      \setlength{\unitlength}{\unitlength * \real{\svgscale}}%
    \fi%
  \else%
    \setlength{\unitlength}{\svgwidth}%
  \fi%
  \global\let\svgwidth\undefined%
  \global\let\svgscale\undefined%
  \makeatother%
  \begin{picture}(1,0.41961616)%
    \lineheight{1}%
    \setlength\tabcolsep{0pt}%
    \put(0,0){\includegraphics[width=\unitlength,page=1]{mdp2.pdf}}%
  \end{picture}%
\endgroup%

%% file: mdp3.pdf_tex
\begingroup%
  \makeatletter%
  \providecommand\color[2][]{%
    \errmessage{(Inkscape) Color is used for the text in Inkscape, but the package 'color.sty' is not loaded}%
    \renewcommand\color[2][]{}%
  }%
  \providecommand\transparent[1]{%
    \errmessage{(Inkscape) Transparency is used (non-zero) for the text in Inkscape, but the package 'transparent.sty' is not loaded}%
    \renewcommand\transparent[1]{}%
  }%
  \providecommand\rotatebox[2]{#2}%
  \newcommand*\fsize{\dimexpr\f@size pt\relax}%
  \newcommand*\lineheight[1]{\fontsize{\fsize}{#1\fsize}\selectfont}%
  \ifx\svgwidth\undefined%
    \setlength{\unitlength}{243.96473802bp}%
    \ifx\svgscale\undefined%
      \relax%
    \else%
      \setlength{\unitlength}{\unitlength * \real{\svgscale}}%
    \fi%
  \else%
    \setlength{\unitlength}{\svgwidth}%
  \fi%
  \global\let\svgwidth\undefined%
  \global\let\svgscale\undefined%
  \makeatother%
  \begin{picture}(1,1.01875096)%
    \lineheight{1}%
    \setlength\tabcolsep{0pt}%
    \put(0,0){\includegraphics[width=\unitlength,page=1]{mdp3.pdf}}%
  \end{picture}%
\endgroup%

%% file: sequential_flow.tex

We introduce first a simple variant of the sequential flow problem.
Our motivations are twofold: first it will be a technically convenient step for our solution of the general problem,
second we provide a complexity lower bound already for this simpler problem.

We let $\st$ be a finite set of states.

\paragraph*{Capacities.}  \AP
A ""capacity"" is an element $\fl a \in (\Nats \cup \set{\omega})^{\st \times \st}$,
intuitively it defines for each pair of "states" $(p,q)$ the maximal value that can be transported from $p$ to $q$ in one step.
Setting $\fl{a}(p,q) = 0$ means that there is no edge from $p$ to $q$.
Using $\omega$ as a value is interpreted as having an edge with unbounded "capacity".

\paragraph*{Flows.}  \AP
Let us call ""flow"" an element $\fl{f} \in \Nats^{\states \times \states}$.
We call ""configuration"" an element of $\Nats^\st$.
A "flow" $\fl{f}$ induces two "configurations" $\intro*\pre(\fl f)$ and $\intro*\post(\fl f)$ defined by 
\[
\pre(\fl f)(p) = \sum_{q \in \st} \fl f(p,q) \qquad \text{and} \qquad \post(\fl f)(q) = \sum_{p \in \st} \fl f(p,q).
\]
Given $\conf{c}, \conf{c'}$ two "configurations" and $\fl{f}$ a "flow", 
we say that $\conf{c}$ ""goes to"" $\conf{c'}$ using $\fl{f}$ and write $\conf c \intro*\flowarrow {\fl f} \conf{c'}$,
if $\conf{c} = \pre(\fl f)$ and $\conf{c'} = \post(f)$. 

For a flow $\fl f$, we write~$\intro*\supp(f) = \set{(p,q) \in \st^2 : \fl f(p,q) > 0}$. 

\paragraph*{Flow "ideals".}  \AP
For a "capacity" $\fl a \in (\Nats \cup \set{\omega})^{\st \times \st}$ we write 
\[
\reintro*\dc{\fl a} = \set{ \fl f \in \Nats^{\st \times \st} : \conf f \leq \fl a},
\]
and call $\dc{\fl a}$ the (flow) "ideal" induced by $\fl a$.
This terminology will be justified later when using the framework of "well quasi orders".

\paragraph*{Sequential "flows".}  \AP
We introduce now an extension of the classical (max-)flow problem with a sequential aspect:
instead of a single "capacity" $\fl a$ we have a finite set of "capacities" $\Capa$,
and we go from a "configuration" $\conf{c}$ to another "configuration" $\conf{c'}$ by applying a sequence of "flows"
$\fl{f_1},\dots,\fl{f_\ell}$ with each flow belonging to some $\dc{\fl{a}}$ with $\fl a \in \Capa$.

\AP Formally, a ""flow word"" is $\bfl{f} = \fl{f_1} \dots \fl{f_{\ell}}$ where each $\fl{f_i}$ is a "flow".
We write $\conf c \intro*\flowpath{\bfl f} \conf{c'}$ if there exists a sequence of "configurations" 
$\conf{c_0}, \conf{c_1}, \dots, \conf{c_\ell}$ such that 
$\conf c = \conf{c_0}$, $\conf c' = \conf{c_\ell}$, and $\conf{c_{i-1}} \flowarrow{f_i} \conf{c_i}$ for all $i \in [1,\ell]$.
In this case, we say that $\conf{c}$ "goes to" $\conf{c'}$ using the "flow word" $\bfl{f}$.

\AP A ""capacity word"" is a finite sequence of "capacities".
For a "flow word" $\bfl{f} = \fl{f_1} \dots \fl{f_{\ell}}$ and a "capacity word" $\fl{w} = \fl{a_1} \dots \fl{a_{\ell}}$ 
we write $\bfl f \le \fl w$ to mean that $\fl f_i \le \fl a_i$ for each position $i$. 

\paragraph*{Configuration ideals.}  \AP
Given a "state"~$q$, we write $\conf{q} \in \Nats^{\states}$ for the vector which has value~$1$ on the~$q$ component and~$0$ elsewhere. 
We use additive notations: for instance, $2 \cdot \conf{q_1} + \conf{q_2}$ has value $2$ in the $q_1$ component, $1$ in the $q_2$ component, and $0$ elsewhere.
For a vector $\fl x \in (\Nats \cup \set{\omega})^\st$ we write
\[
\dc{\fl x} = \set{ \fl c \in \Nats^{\st} : \conf c \leq \fl x},
\]
and call $\dc{\fl x}$ the (configuration) "ideal" induced by $\fl x$.

\paragraph*{The "simple sequential flow problem".}  \AP
Let us first introduce the simple variant of the "sequential flow problem" and give an example.

\begin{problem}[""Simple sequential flow problem""]\hfill

\textbf{Input}: $\states$ a finite set of "states", $\Capa$ a finite set of "capacities", $s$ a source state, and $\fin$ a target ideal of configurations.

\AP 
\textbf{Output}: Yes if for all $n \in \Nats$, there exists a "capacity word" $w \in \Capa^*$, a "flow word" $\bfl f \leq w$ and a final capacity $c \in \fin$ such that $n \cdot s \flowpath{\bfl f} c$, and No otherwise.
\phantomintro\SSFP
\end{problem}

When considering an instance of the "simple sequential flow problem" 
we let $\decFlows = \bigcup_{a \in \Capa} \dc{a}$ denote the downward closed set of flows,
$I = \dc{(\omega \cdot s)}$, and $\fin = \dc{(\omega \cdot t)}$.
We call $I$ the source "ideal" and $\fin$ the target "ideal".

In the instance of the "simple sequential flow problem" represented in Figure~\ref{fig:flow1}, we ask the following question: 
can $\fin = \dc{(\omega \cdot q_4)}$ be reached from any "configuration" of $I = \dc{(\omega \cdot \conf{q_1})}$
using the capacities $\Capa = \set{a,b,c}$? 
The answer is yes: for every $n$, there exists a "flow word" $\bfl f \le (\fl a \fl c^{n-1} \fl b)^n$ such that
$n \cdot \conf{q_1} \flowpath{\bfl f} n \cdot \conf{q_4} \in \fin$.
The "flow word" $\bfl f$ is represented in Figure~\ref{fig:flow2}.

\begin{figure}[!ht]
    \centering
    \def\svgwidth{\columnwidth}
    \resizebox{0.8\textwidth}{!}{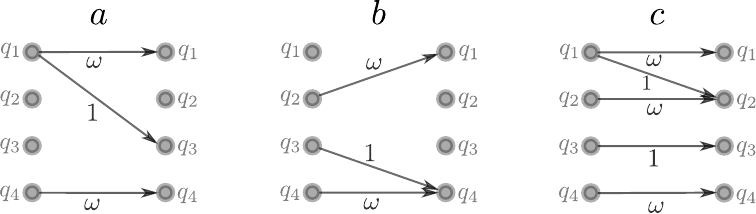}
    \caption{An instance of the "simple sequential flow problem". 
    Can $\fin = \dc{(\omega \cdot q_4)}$ be reached from any "configuration" of $I = \dc{(\omega \cdot \conf{q_1})}$
    using the three "capacities" $\fl a$, $\fl b$, and $\fl c$?
    }
    \label{fig:flow1}
\end{figure}

\begin{figure}[!ht]
    \centering
    \def\svgwidth{\columnwidth}
    \resizebox{\textwidth}{!}{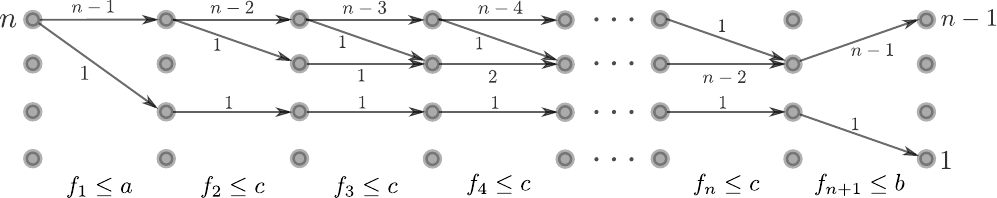}
    \caption{A "flow word" $\bfl f = \fl{f_1} \fl{f_2} \dots \fl{f_{n+1}} \leq \fl{a} \fl{c}^{n-1} \fl b$ 
    such that $n \cdot \conf{q_1}$ "goes to" $(n-1) \cdot \conf{q_1} + \conf{q_4}$ using $\bfl f$. 
    This construction can be iterated to construct $\bfl f \leq \fl (\fl{a} \fl{c}^{n-1} \fl b)^n$ such that 
    $n \cdot \conf{q_1}$ "goes to" $n \cdot \conf{q_4}$ using~$\bfl f$.}
    \label{fig:flow2}
\end{figure}

The general "sequential flow problem" extends its simple variant by considering more general sets of "configurations";
in the "simple sequential flow problem" we ask whether all "configurations" of $I = \dc{(\omega \cdot s)}$ 
go to some "configuration" in $\fin = \dc{(\omega \cdot t)}$, in the general "sequential flow problem" 
we consider an arbitrary "downward closed" set $F$ and ask to compute the ("downward closed") set of all "configurations" reaching $F$.
Towards defining the "sequential flow problem" let us recall some classical terminology about "well quasi orders", which allow us to manipulate general "downward closed" sets (\cite{HG52,K72}, see~\cite{S17} for an exposition of recent results).

\paragraph*{Well quasi orders.}  \AP
Let $(E,\leq)$ be a quasi ordered set (\textit{i.e.} $\leq$ is reflexive and transitive), 
it is a ""well quasi ordered set"" ("WQO") if any infinite sequence contains a non-decreasing pair. 
A set $S$ is ""downward closed"" if for any $x \in S$, if $y \leq x$ then $y \in S$.
An ""ideal"" is a non-empty "downward closed" set $I \subseteq E$ such that for all $x,y \in I$, 
there exists some $z \in I$ satisfying both $x \leq z$ and $y \leq z$. 

\AP We equip the set of "configurations" $\Nats^{\states}$ and the set of "flows" $\Nats^{\states \times \states}$ 
with the quasi orders $\leq$ defined component wise;
thanks to Dickson's Lemma~\cite{D13} they are both "WQOs".

\begin{lemma}\label{lem:WQOstuff}[Folklore, see~\cite{S17}]
Let $(E,\leq)$ be a "WQO".
\begin{itemize}
\itemAP Any infinite sequence of non-increasing "downward closed" sets is eventually constant.
\itemAP A subset is "downward closed" if and only if it is a finite union of incomparable "ideals".
We call it its ""decomposition into ideals"" (or simply, its \reintro{decomposition}), which is unique (up to  permutation).
\itemAP An "ideal" is included in a "downward closed set" if and only if it is included in one of the "ideals" of its "decomposition".
\end{itemize}

Let $X$ be a finite set, we consider the "WQO" $(\Nats^\X, \le)$ with $\le$ defined component wise.
A subset of~$\Nats^\X$ is an "ideal" if and only if it is of the form\phantomintro{\dc}
  \[
  \reintro*\dc{\conf a} = \{ \conf c \in \Nats^\X \mid \conf c \leq \conf a \},
  \]
for some $\conf a \in (\Nats \cup \set{\omega})^X$.
\end{lemma}
Thanks to Lemma~\ref{lem:WQOstuff}, we represent "downward closed sets" of "configurations" and "flows" using their decomposition into ideals.

The most general definition of the "sequential flow problem" is the following.
Let $\states$ be a finite set of states, 
$\Capa$ a finite set of "capacities" inducing $\intro*\decFlows \subseteq \Nats^{\states \times \states}$ a "downward closed" set of "flows",
$\fin \subseteq \Nats^{\states}$ a downward closed set of configurations,
and $\fin \subseteq \Nats^{\states}$ a "downward closed" set of "configurations",
let us define the following "downward closed set" $\SFP$:
\begin{align*}
   \SFP(\decFlows,\fin) = \{\conf c\in\Nats^{\states}\mid \conf{c} \flowpath{\bfl f}\conf{c'}\in\fin,\ \bfl f\in\decFlows^*\}\ ,
\end{align*} 
\textit{i.e.} the "configurations" from which one may reach $\fin$ using only "flows" from $\decFlows$.
The objective is to construct $\SFP(\decFlows,\fin)$.

The following classical result of~\cite{VJ85} allows us to further reduce our problem.
Informally, it says that the task of constructing a downward closed set can be reduced to the task of 
determining whether an "ideal" is included in a "downward closed" set.

\begin{lemma}\label{lem:vjgl}
Let $X$ be a downward closed set.
Assume there exists an algorithm solving the following problem:

\textbf{Input}: $I$ ideal

\textbf{Output}: Yes if $I \subseteq X$, and No otherwise.

\noindent Then there is an algorithm constructing the decomposition of $X$ into ideals.
\end{lemma}

Noting that $\SFP(\decFlows, \fin_1 \cup \fin_2) = \SFP(\decFlows, \fin_1) \cup \SFP(\decFlows, \fin_2)$,
we can without loss of generality assume that $\fin$ is a single ideal.
We therefore define the "sequential flow problem" as follows.

\begin{problem}[""Sequential flow problem""]\hfill

\AP
\textbf{Input}: $\states$ a finite set of "states", $\Capa$ a finite set of "capacities", $a$ a source capacity, and $b$ a target capacity.
We let $\decFlows = \bigcup_{c \in \Capa} \dc{c}$ denote the downward closed set of flows induced by $\Capa$,
$I = \dc{a}$ the source "ideal", and $\fin = \dc{b}$ the target "ideal".

\textbf{Output}: Yes if $I \subseteq \SFP(\decFlows,\fin)$, meaning whether $F$ can be reached from all "configurations" of $I$ using only "flows" from $\decFlows$, and No otherwise.
\end{problem} 

As argued above, an algorithm solving the sequential flow problem
actually yields an algorithm constructing the ideal decomposition of $\SFP(\decFlows,\fin)$.

%% file: flows1.pdf_tex
\begingroup%
  \makeatletter%
  \providecommand\color[2][]{%
    \errmessage{(Inkscape) Color is used for the text in Inkscape, but the package 'color.sty' is not loaded}%
    \renewcommand\color[2][]{}%
  }%
  \providecommand\transparent[1]{%
    \errmessage{(Inkscape) Transparency is used (non-zero) for the text in Inkscape, but the package 'transparent.sty' is not loaded}%
    \renewcommand\transparent[1]{}%
  }%
  \providecommand\rotatebox[2]{#2}%
  \ifx\svgwidth\undefined%
    \setlength{\unitlength}{362.609bp}%
    \ifx\svgscale\undefined%
      \relax%
    \else%
      \setlength{\unitlength}{\unitlength * \real{\svgscale}}%
    \fi%
  \else%
    \setlength{\unitlength}{\svgwidth}%
  \fi%
  \global\let\svgwidth\undefined%
  \global\let\svgscale\undefined%
  \makeatother%
  \begin{picture}(1,0.28204253)%
    \put(0,0){\includegraphics[width=\unitlength,page=1]{flows1.pdf}}%
  \end{picture}%
\endgroup%

%% file: flows2.pdf_tex
\begingroup%
  \makeatletter%
  \providecommand\color[2][]{%
    \errmessage{(Inkscape) Color is used for the text in Inkscape, but the package 'color.sty' is not loaded}%
    \renewcommand\color[2][]{}%
  }%
  \providecommand\transparent[1]{%
    \errmessage{(Inkscape) Transparency is used (non-zero) for the text in Inkscape, but the package 'transparent.sty' is not loaded}%
    \renewcommand\transparent[1]{}%
  }%
  \providecommand\rotatebox[2]{#2}%
  \ifx\svgwidth\undefined%
    \setlength{\unitlength}{478.11699046bp}%
    \ifx\svgscale\undefined%
      \relax%
    \else%
      \setlength{\unitlength}{\unitlength * \real{\svgscale}}%
    \fi%
  \else%
    \setlength{\unitlength}{\svgwidth}%
  \fi%
  \global\let\svgwidth\undefined%
  \global\let\svgscale\undefined%
  \makeatother%
  \begin{picture}(1,0.19864197)%
    \put(0,0){\includegraphics[width=\unitlength,page=1]{flows2.pdf}}%
  \end{picture}%
\endgroup%

%% file: reduction.tex

\AP Let us consider an "MDP" $\intro*\redmdp$ and a "target" $t \in \st$.
We first recall a folklore result reducing the almost sure reachability question for "MDPs" to solving a two player B{\"u}chi game 
(we refer to~\cite{GTW2002} for the definitions and notations of B{\"u}chi games).
The B{\"u}chi game is played between \intro*\Eve and \intro*\Adam as follows. 
From a state $p$:
\begin{enumerate}
\itemAP \Eve chooses an "action" $a$ and a "transition" $(p,q) \in \Delta_a$;
\itemAP \Adam can either choose to
	\begin{description}
    \itemAP[""agree""] and the game continues from $q$, or
    \itemAP[""interrupt""] and choose another "transition" $(p,q') \in \Delta_a$, the game continues from $q'$.
    \end{description}
\end{enumerate}
We say that Eve wins if either the target state $t$ is reached or \Adam interrupts infinitely many times.
This is indeed a B{\"u}chi objective where the B{\"u}chi transitions are the interruptions of Adam and the self-loop around $t$.

\begin{lemma}\label{lem:Buchi}
There exists a "strategy" ensuring almost surely to reach $t$ from $s$ 
if and only if \Eve has a winning "strategy" from $s$ in the above B{\"u}chi game.
\end{lemma}

\AP We now explain how this reduction can be extended to the "stochastic control problem".
Let us consider an "MDP" $\redmdp$ and a "target" $t \in \st$.
We define an infinite B{\"u}chi game $\intro*\BuchiG$. 
The set of vertices is the set of "configurations" $\Nats^{\st}$.
The game is played as follows from a "configuration"~$\conf c$: 
\begin{enumerate}
\itemAP \Eve chooses an "action" $a$ and a "flow" $\fl f$ such that $\pre(\fl f) = \conf c$
and $\supp (\fl f) \subseteq \Delta_a$.
\itemAP \Adam can either choose to
    \begin{description}
      \itemAP[\emph{agree}] and the game continues from $\conf c' = \post(\fl f)$
      \itemAP[\emph{interrupt}] and choose a "flow" $\fl f'$ such that $\pre(\fl f') = \conf c$
    and $\supp (\fl f') \subseteq \Delta_a$, and the game continues from $\conf{c''} = \post(\fl f')$.
    \end{description}
\end{enumerate}
Note that \Eve choosing a flow $\fl f$ is equivalent to choosing for each "token" a "transition" $(p,q) \in \Delta_a$, 
inducing the "configuration" $\conf{c'}$, and similarly for Adam should he decide to "interrupt".

We now define the winning objective: \Eve wins if either all "tokens" are in the target "state", or if \Adam "interrupts" infinitely many times. Again this is a B{\"u}chi objective.

Although the game is infinite, it is actually a disjoint union of finite games.
Indeed, along a play the number of "tokens" is fixed, so all visited configurations belong to $\st^n$ for some $n \in \Nats$.

\begin{lemma}\label{lem:reduction-correctness}
Let $\conf{c}$ be a configuration with $n$ tokens in total, the following are equivalent:
\begin{itemize}
\item There exists a strategy almost surely reaching $n \cdot t$ from $\conf{c}$,
\item \Eve has a winning strategy in the B{\"u}chi game $\BuchiG$ starting from~$\conf c$.
\end{itemize}
\end{lemma}
Lemma~\ref{lem:reduction-correctness} follows from applying Lemma~\ref{lem:Buchi} on the product "MDP" $\redmdp^n$.

\AP We define the game $\intro*\BuchiG^i$ for $i \in \Nats$, which is defined just as $\BuchiG$ except for the winning objective: 
\Eve wins in $\BuchiG^i$ if either all tokens are in the target state, or if \Adam "interrupts" more than $i$ times. 
It is clear that if \Eve has a winning strategy in $\BuchiG$ then she has a winning strategy in $\BuchiG^i$. 
Conversely, the following result states that $\BuchiG^i$ is equivalent to $\BuchiG$ for some $i$.
Let $\fin$ be the set of "configurations" for which all tokens are in "state"~$t$. 
We let $\Win^i \subseteq \Nats^{\st}$ denote the winning region for \Eve in the game $\BuchiG^i$,
and $\Win^{\infty} \subseteq \Nats^{\st}$ denote the winning region for \Eve in the game $\BuchiG$.

\begin{lemma}\label{lemma:finiteindex}\hfill
\begin{itemize}
	\item If $\Win^i = \Win^{i+1}$, then $\Win^i = \Win^{i + j}$ for all $j \ge 0$,
	\item $(\Win^i)_{i \in \Nats}$ is a non-increasing sequence of downward closed sets,
	\item There exists $i \in \Nats$ such that $\Win^i = \Win^{\infty}$.
\end{itemize}
\end{lemma}

\begin{proof}
The first item follows from the definition of $\Win^i$.

For the second item, we argue that each $\Win^i$ is "downward closed": 
if $\conf c \le \conf c'$ and $\conf c' \in \Win^i$, then $\conf c \in \Win^i$.
Let $\sigma$ be a winning strategy from $\conf c'$, we show that it induces a strategy from $\conf c$:
for any play from $\conf c'$ ending in $\conf c'_e$ there exists a play from $\conf c$ ending in $\conf c_e$ consistent with $\sigma$ such that $\conf c_e \le \conf c'_e$ and the two plays are interrupted at the same steps.
This is easily proved by induction on the length of the play.
Since $\sigma$ is a winning strategy from $\conf c'$ all plays either reach $n \cdot t$ for some $n \in \Nats$ or Adam interrupts $i$ times, so there exists a length $N$ such that either happens within the first $N$ steps.
The invariant above implies that the corresponding strategy is winning from $\conf c$.

For the third item, let $X = \bigcap_{i \in \Nats} \Win^i$, we first argue that $X = \Win^{\infty}$.
It is clear that $\Win^{\infty} \subseteq X$: 
if \Eve has a strategy to ensure that either all tokens are in the target state, or that \Adam "interrupts" infinitely many times, then it particular this is true for Adam interrupting more than $i$ times for any $i \in \Nats$.
For the converse inclusion, we show that $\Nats^{\st} \setminus \Win^{\infty} \subseteq \Nats^{\st} \setminus X$:
let $\conf c$ such that Adam has a winning strategy in $\BuchiG$ from $\conf c$. 
By positional determinacy of B{\"u}chi games, we can choose the strategy to be positional.
Let $n$ the number of tokens in $c$, and $i > n \cdot |\st|$, 
we claim that Adam's strategy ensures never to interrupt more than $i$ times from $c$.
Indeed if there exists a play from $c$ consistent with that strategy where he interrupts more than $i$ times, 
then a simple pumping argument yields a play from $c$ consistent with the strategy interrupting infinitely many times, a contradiction.
Thus $\conf c$ is in $\Nats^{\st} \setminus \Win^i$, hence a fortiori in $\Nats^{\st} \setminus X$.

Thanks to the second item $(\Win^i)_{i \ge 0}$ is a non-increasing sequence of "downward closed sets" in $\Nats^{\st}$,
so it stabilises thanks to Lemma~\ref{lem:WQOstuff}, \textit{i.e.} there exists $i_0 \in \Nats$ such that $\Win^{i_0} = \bigcap_i \Win^i$, 
which concludes the proof of the third item.
\end{proof}

Note that Lemma~\ref{lem:reduction-correctness} and Lemma~\ref{lemma:finiteindex} substantiate the claims made in Section~\ref{sec:problem}:
pure positional strategies are enough and the answer to the stochastic control problem does not depend upon the exact probabilities in the "MDP".
Indeed, the construction of the B{\"u}chi games does not depend on them, 
and the answer to the former is equivalent to determining whether \Eve has a winning strategy in each of them.

We are now fully equipped to show that a solution to the "sequential flow problem"
yields the decidability of the "stochastic control problem".
\AP Note first that $\Win^0 = \SFP(\Flows^0,\fin)$ where\phantomintro{\Flows^{0}}
\begin{align*}
  \reintro*\Flows^0=\{\fl f \in \Nats^{\st \times \st} \mid \exists a \in \act,\ \supp (\fl f) \subseteq \Delta_a\}.
\end{align*}
Indeed, in the game $\BuchiG^0$ \Adam cannot "interrupt" as this would make him lose immediately. 
Hence, the winning region for \Eve in $\BuchiG^0$ is
$\SFP(\Flows^0,\fin)$.
\AP We generalise this by setting $\intro*\Flows^{i}$ for all~$i>0$ to be the set of "flows"~$\fl f\in\Nats^{\st\times\st}$ such that for some "action"~$a \in \actions$,
\begin{itemize}
\item $\supp(\fl f)\subseteq\transitions_a$, and
\item for every~$\fl{f'}$, if $\pre(\fl{f'}) = \pre(\fl f)$ and $\supp(\fl{f'}) \subseteq \transitions_a$,
then $\post(\fl{f'}) \in \Win^{i-1}$.
\end{itemize}
Equivalently, this is the set of "flows" for which, when played in the game~$\BuchiG$ by~\Eve, \Adam cannot use an "interrupt" move and force the configuration outside of~$\Win^{i-1}$.

\begin{lemma}\label{lem:characterisation_winning_region}
For all $i \in \Nats$, 
\[
\Win^{i} = \SFP(\Flows^i,\fin)
\]
\end{lemma}

Before proving this key lemma, let us discuss how it induces an algorithm for solving the "stochastic control problem".
Let us write $I = \dc{(\omega \cdot s)}$ for the source ideal.
Thanks to Lemma~\ref{lem:reduction-correctness} the answer to the stochastic control problem is Yes
if and only if $I \subseteq \Win^{\infty}$.
Thanks to Lemma~\ref{lemma:finiteindex} to compute $\Win^{\infty}$,
it is enough to compute $\Win^0, \Win^1, \dots$ until $\Win^i = \Win^{i+1}$, implying that $\Win^i = \Win^{\infty}$.
Thanks to Lemma~\ref{lem:characterisation_winning_region}, 
for each~$i$ computing~$\Win^i$ reduces to solving the sequential flow problem.

\begin{proof}
We proceed by induction on $i$.

Let~$\conf c$ be a "configuration" in $\SFP(\Flows^i,\fin)$. This means that
there exists a "flow word"~$\bfl f = \fl {f_1} \cdots \fl {f_\ell}$ such that~$f_k \in \Flows^i$ for all~$k$,
and $\conf c\flowpath{\bfl f}\conf{c'} \in \fin$. 
Expanding the definition, there exist $\conf c_0 = \conf c, \dots, \conf c_\ell = \conf c'$
such that $\conf c_{k-1} \flowarrow{f_k} \conf c_k$ for all~$k$.

Let us now describe a strategy for \Eve in~$\BuchiG^i$ starting from~$\conf c$.
As long as \Adam "agrees", \Eve successively chooses the sequence of flows $f_1,f_2,\dots$
and the corresponding configurations $\conf c_1,\conf c_2,\dots$.
If~\Adam never "interrupts", then the game reaches the "configuration" $\conf{c'} \in \fin$, and \Eve wins.
Otherwise, as soon as \Adam "interrupts", 
by definition of~$\Flows^i$, we reach a "configuration" $\conf d \in \Win^{i-1}$.
By induction hypothesis, \Eve has a strategy which ensures from~$\conf d$ to 
either reach~$\fin$ or that \Adam interrupts at least~$i-1$ times. 
In the latter case, adding the "interrupt" move leading to $\conf d$ yields at least~$i$ "interrupts",
so this is a winning strategy for \Eve in~$\BuchiG^i$, witnessing that~$\conf c \in \Win^i$.

Conversely, assume that there is a winning strategy~$\sigma$ of~\Eve in~$\BuchiG^i$ from a "configuration"~$\conf c$.
Consider a play consistent with $\sigma$, it either reaches~$\fin$ or \Adam "interrupts".
Let~$\bfl f = f_1 f_2 \dots f_{\ell}$ denote the sequence of flows until then.
We argue that~$f_k \in \Flows^i$ for~$k \in [1,\ell]$. 
Let $f = f_k$ for some $k$, by definition of the game~$\supp(f) \subseteq \transitions_a$ for some "action"~$a$. 
Let~$\fl f'$ such that~$\pre(\fl{f'}) = \pre(\fl{f})$ and~$\supp(\fl{f'}) \subseteq \transitions_a$. 
In the game~$\BuchiG$ after \Eve played~$\fl{f_k}$, 
\Adam has the possibility to "interrupt" and choose~$\fl{f'}$.
From this configuration onward the strategy $\sigma$ is winning in~$\BuchiG^{i-1}$,
implying that~$\fl{f} \in \Flows^i$.
Thus~$\bfl f=\fl{f_1}\fl{f_2}\dots\fl{f_{\ell}}$ is a witness that~$\conf c \in \Win^i$.
\end{proof}

%% file: decidability_simple.tex
We study the "simple sequential flow problem", stated as follows.
Let $\states$ be a finite set of states, $\Flows \subseteq \Nats^{\states \times \states}$ a "downward closed" set of "flows" represented as the finite union of capacity ideals from $\Capa$, $s$ a source state inducing $I = \dc {(\omega \cdot s)}$ the source "ideal", and $\fin = \dc b$ the target "ideal".
Determine whether for all $n \in \Nats$, there exists a "configuration" $c \in \fin$, a "capacity word" $w \in \Capa^*$, and a "flow word" $\bfl f \leq w$ such that $n \cdot s \flowpath{\bfl f} c$.

\begin{theorem}\label{thm:main-simple-sfp}
The "simple sequential flow problem" is decidable in $\textsc{EXPSPACE}$.
\end{theorem}

The proof makes use of two ingredients: "distance automata" and the max-flow min-cut theorem.

\medskip
\AP ""Distance automata"" are weighted automata over the semiring $(\Nats \cup \{\omega\},+,\min)$:
a weighted automaton over the alphabet $\Sigma$ is given by a set of states $\st$, an initial state $q_0 \in \st$,
a vector of final costs $\eta \in (\Nats \cup \{\omega\})^\st$, and a transition relation
\[
\Delta \subseteq \st \times \Sigma \times (\Nats \cup \{\omega\}) \times \st,
\]
with the following interpretation: $(p,a,c,q) \in \Delta$ is a transition from $p$ to $q$ reading letter $a$ with cost $c$.
A run $\rho$ over $w = w_1 \dots w_\ell \in \Sigma^*$ is a sequence of transitions
\[
\rho = (q_0,w_1,c_1,q_1) (q_1, w_2, c_2, q_2) \dots (q_{\ell - 1}, w_\ell, c_\ell, q_\ell).
\]

The ""value"" of $\rho$ is $\val(\rho) = \sum_{i \in [1,\ell]} c_i + \eta(q_\ell)$.
A weighted automaton $\aut$ induces the function $\sem{\aut} : \Sigma^* \to \N \cup \set{\omega}$ defined by
\[
\sem{\aut}(w) = \min_{\rho \text{ run over } w} \val(\rho).
\]
We say that $\aut$ recognises the function $\sem{\aut}$.

The ""boundedness problem"" for "distance automata" asks whether for a given "distance automaton" $\aut$ the function $\sem{\aut}$ is bounded, meaning:
\[
\exists n \in \Nats,\ \forall w \in \Sigma^*,\ \sem{\aut}(w) \le n.
\]

\begin{thmC}[\cite{Kirsten05}]\label{thm:boundedness-of-distance}
The "boundedness problem" for "distance automata" is decidable in $\textsc{PSPACE}$.
\end{thmC}

We reduce the "simple sequential flow problem" to the "boundedness problem" for "distance automata" by considering the following function:
\[
\begin{array}{lccl}
\Phi: & \Capa^* & \longrightarrow & \Nats \cup \{\omega\} \\
      & \fl w   & \longmapsto     & \sup \{n \in \Nats \mid \exists c \in \fin, \exists \bfl f, 
      \bfl f \leq \fl w \text{ and } n \cdot s \flowpath{\bfl f} c\}.
\end{array}
\]
In words, given as input a "word of capacities" $w$, $\Phi$ computes the largest number of "tokens" one may synchronize in $F$ using $w$.

By construction, whether for all $n \in \Nats$, there exists a "configuration" $c \in \fin$, a "capacity word" $w \in \Capa^*$ and a "flow word" $\bfl f \leq w$ such that $n \cdot s \flowpath{\bfl f} c$ is equivalent to $\Phi$ being unbounded. We will prove the following.

\begin{lemma}\label{lem:distance}
The function $\Phi$ is computed by a "distance automaton" with $2^{|\states|}$ "states".
\end{lemma}

Theorem~\ref{thm:main-simple-sfp} follows by combining Theorem~\ref{thm:boundedness-of-distance} and Lemma~\ref{lem:distance}. 
The construction of a distance automaton for Lemma~\ref{lem:distance} relies on the standard max-flow min-cut theorem, which we expose now.

\paragraph*{Max-flow min-cut theorem.} \AP
In the classical (max-)flow problem the input is a "capacity" $\fl a \in (\Nats \cup \set{\omega})^{V \times V}$, a source vertex $s$, and a target vertex $t$.
The capacity $\fl a$ defines a graph over the set of states $V$: there is an edge from $u$ to $v$ 
if $\fl{a}(u,v) \neq 0$.
The objective is to compute the maximal "$(s,t)$-flow" respecting the "capacity" $\fl a$.
An ""$(s,t)$-flow"" is an element $\fl{f} \in \Nats^{V \times V}$ such that for all vertices $v$ except $s$ and $t$
we have a ""conservation law"":
\[
\sum_{v' \in V} f(v',v) = \sum_{v' \in V} f(v,v').
\]
The value of an "$(s,t)$-flow" is $|f| = \sum_{v \in V} f(s,v)$, which thanks to "conservation laws"
is also equal to $\sum_{v \in V} f(v,t)$.
We say that the "flow" $\fl{f}$ respects the "capacity" $\fl{a}$ if for all $v,v' \in V$ we have $\fl{f}(v,v') \le \fl{a}(v,v')$.

An ""$(s,t)$-cut"" is a set of edges $C \subseteq V \times V$ such that removing them disconnects $t$ from $s$. 
The ""cost of a cut"" is the sum of the weight of its edges: $|C| = \sum_{(v,v') \in C} \fl{a}(v,v')$.
The ""max-flow min-cut theorem"" states that the maximal value of an "$(s,t)$-flow" is exactly the minimal cost of an "$(s,t)$-cut" (\cite{FF56}). With this in hands, we are now ready to prove Lemma~\ref{lem:distance}.

\begin{proof}[Proof of Lemma~\ref{lem:distance}]
Let us fix a "capacity word" $\fl w = \fl w_1 \dots \fl w_\ell \in \Capa^*$.
Consider the finite graph $\G_w$ with vertex set $(\states \times [0,\ell]) \cup \{t\}$, where $t$ is an additionnal fresh vertex, and for all $i \in [1,\ell]$, 
an edge from $(q,i-1)$ to $(q',i)$ labelled by $\fl w_i(q,q')$, and an edge labelled by $b(q)$ from $(q, \ell)$ to $t$.
Then $\Phi(\fl w)$ is the maximal value of an "$((s,0),t)$-flow" in $\G_w$.
Thanks to the max-flow min-cut theorem, this is also the minimal cost of an "$((s,0),t)$-cut" in $\G_w$.

The "distance automaton" we construct guesses an $((s,0),t)$-cut and outputs its cost,
hence its value on a given word corresponds to the minimal cost "cut". 
To verify that the current run indeed corresponds to a "cut", it suffices to remember the set of reachable "states" from $(s,0)$ using a powerset construction. In the last layer, to obtain a "cut", one is forced to "cut" the edge from $(q,\ell)$ to $t$, which has cost $b(q)$, for all states $q$ such that $(q,\ell)$ is reachable from $(s,0)$.

The set of "states" of $\aut$ is $\pow{\states}$, the initial state is $\set{s}$, and the vector of final costs $\eta \in (\Nats \cup \{\omega\})^{\pow\states}$takes value $\sum_{q \in X} b(x)$ on subset $X \in \pow{\states}$.
The transition relation $\Delta$ is defined as follows: $(X, a, c, Y) \in \Delta$ 
if $Y \subseteq Z = \set{q' \in \states : \exists q \in X, a(q,q') \neq 0}$
and $c = \sum_{q \in X} \sum_{q' \in Z \setminus Y} a(q,q')$.
In words: from the set $X$ representing the set of reachable states from $(s,0)$ without the edges already cut,
the new set $Y$ is a subset of the set of reachable states from $X$ reading $a$,
obtained by removing a set of edges whose cost is accounted for in the current value of the run.

A run of the automaton $\aut$ corresponds to an "$((s,0),t)$-cut", and its value is the cost of the cut. 
Minimising over all runs yields the minimal cost of an "$((s,0),t)$-cut", so $\sem{\aut} = \Phi$ by the max-flow min-cut theorem.
\end{proof}

%% file: decidability.tex
This section builds on Lemma~\ref{lem:distance} in order to obtain decidability of the general "sequential flow problem". 
Let $\states$ be a finite set of "states", $\Capa$ a finite set of "capacities", $a$ a source capacity and $b$ a target capacity.
We let $\decFlows = \bigcup_{a \in \Capa} \dc{a}$ denote the downward closed set of flows,
$I = \dc{a}$ the source "ideal", and $\fin = \dc{b}$ the target "ideal".
We ask whether $I \subseteq \SFP(\Flows, \fin)$, that is, whether $\fin$ can be reached from any "configuration" of $I$ using only "flows" from $\Flows$.

\begin{theorem}\label{thm:main-sfp}
The "sequential flow problem" is decidable.
\end{theorem}

To prove Theorem~\ref{thm:main-sfp}, we will use so-called ""desert automata"". 
Dually to "distance automata", these non-deterministic automata output the maximal value over all runs. Formally, "desert automata" are weighted automata over the semi-ring $(\Nats \cup \{\omega\}, +, \max)$. They are defined as tuples $(\st, q_0, \delta, \eta)$, just as "distance automata" (see Section~\ref{sec:decidability_simple}), except for their semantics which is now given by
\[
\sem{\aut}(w) = \max_{\rho \text{ run over } w} \val(\rho).
\]
We say that $\aut$ recognises the function $\sem{\aut}$.

\AP \textbf{The theory of regular cost functions.}
The following duality theorem is crucial to our proof and is a central result of the theory of ""regular cost functions""~\cite{Colcombet13}, which is a set of tools for solving boundedness questions.
Since in the "theory of regular cost functions", when considering functions we are only interested in whether they are bounded or not, 
we consider functions ``up to boundedness properties''.
Concretely, this means that a ""cost function"" is an equivalence class of functions $\Sigma^* \to \Nats \cup \set{\omega}$,
with the equivalence being $f \intro*\cfeq g$ if there exists $\alpha : \Nats \to \Nats $ such that $f(w)$ is finite if and only if~$g(w)$ is finite, and in this case, $f(w)\leq \alpha(g(w))$ and  $g(w)\leq \alpha(f(w))$. This is equivalent to
stating that for all~$X\subseteq \Sigma^*$, $f$ is bounded over~$X$ if and only if~$g$ is bounded over~$X$.

\begin{thmC}[\cite{Colcombet13}]
\label{thm:distance-to-desert}
Let $\aut'$ be a "distance automaton".
There exists a "desert automaton" $\aut$ recognising the cost function $\sem{\aut'}$.
\end{thmC}

The boundedness problem for "desert automata" is decidable.

\begin{thmC}[\cite{Kirsten05}]\label{thm:boundedness-of-desert}
The "boundedness problem" for "desert automata" is decidable in $\textsc{PSPACE}$.
\end{thmC}

Our goal is to determine the boundedness of the function
\[
\begin{array}{lccl}
\Phi: & \Capa^* & \longrightarrow & \Nats \cup \{\omega\} \\
      & \fl w   & \longmapsto     & \sup \{n \in \Nats \mid \exists \bfl f, \bfl f \leq \fl w \text{ and } a_n \flowpath{\bfl f} F\}.
\end{array}
\]

\begin{lemma}\label{lem:desert}
The function $\Phi$ is computed by a "desert automaton".
\end{lemma}

\begin{proof}
We let $S = \{s \in \st \mid a(s)= \omega\}$, and for $n \in \Nats \cup \{\omega\}$, let $a_n$ denote the vector obtained from $a \in (\Nats \cup \{\omega\})^\st$ by replacing occurrences of $\omega$ by $n$. Note that $a = a_0 + \sum_{s \in S} \omega \cdot s$.
The automaton recognising $\Phi$ will non-deterministically decompose a flow into synchronous flows:
\begin{itemize}
	\item a finitary part bringing $a_0$ to the target, and
	\item for each $s \in S$, an instance of the "simple sequential flow problem" with source $s$.
\end{itemize}

We now formalise the notion of decomposition.
Recall that $\Capa$ is a finite set of capacities, let $k = \max \set{n \in \Nats : \exists a \in \C, q \in \st, a(q) = n}$
the largest integer appearing in $\Capa$: then $\Capa \subseteq ([0,k] \cup \set{\omega})^\st$.

Let $\Jc = S \cup \{\finite\}$, we consider the alphabets 
\[
\Sigma_2 = \left( \left( [0,k] \cup \set{\omega} \right)^{\st \times \st} \right)^{\Jc}\ ;\ 
\Sigma_1 = \left( \left( [0,k] \cup \set{\omega} \right)^{\st} \right)^{\Jc} \ ;\
\Sigma = \Sigma_2 \cup \Sigma_1.
\]
For a letter $a^{\Jc} = (a^{j})_{j \in \Jc}$ in $\Sigma_2$ we say that $a^{\Jc}$ is a ""decomposition"" of $a \in \Capa$
if $\sum_{j \in \Jc} a^{j} = a$, and similarly a letter $b^{\Jc} \in \Sigma_1$ is a "decomposition" of the target ideal $b$ if $\sum_{j \in \Jc} b^{j} = b$.
We extend "decompositions" to words by considering each position: 
a word $w^{\Jc} \in \Sigma_w^*$ is a "decomposition" of $w \in \C^*$ if for each position $i$,
the letter $w^{\Jc}_i$ is a "decomposition" of $w_i$.

We start with constructing the following automata.
\begin{itemize}
	\item the "distance automaton" $\aut_\finite$ has semantics
\[
\sem{\aut_\finite}(w^{\text{fin}} b^{\text{fin}}) = \begin{cases}
\omega \text{ if } \exists \bfl f,\ \bfl f \leq w^{\text{fin}} \text{ and } a_0 \flowpath{\bfl f} b^{\text{fin}}, \\
0 \text{ otherwise}.
\end{cases}
\]
Every transition has cost $0$ or $\omega$; the automaton non-deterministically guesses the path of each of the finitely many tokens.
The set of states is $\Nats^{p}$ where $p = \sum_{q \in \states} a_0(q)$.

	\item For each $s \in S$, the "distance automaton" $\aut_s$ computes
\[
\sem{\aut_s}(w^{s} b^{s}) = \sup \set{n \in \Nats \mid \exists \bfl f,\ \bfl f \leq w^{s}  \text{ and }  n \cdot s \flowpath{\bfl f} b^{s}},
\]
and its existence is guaranteed by Lemma~\ref{lem:distance}.
\end{itemize}
Note that 
\[
\Phi(w) = \max_{\substack{w^{\Jc}, b^{\Jc} \text{ decompositions} \\ \text{of } w, b}} \min_{j \in \Jc}\ \sem{\aut_j}(w^j b^j).
\]

To obtain a "desert automaton" recognising the cost function $\Phi$, we need three more steps.
The first is to construct a "distance automaton" $\aut'$ computing
\[
\sem{\aut'}(w^{\Jc} b^{\Jc}) = \min_{j \in \Jc}\ \sem{\aut_j}(w^{j} b^{j}),
\]
which is obtained as a direct product of the automata $(\aut_j)_{j \in \Jc}$.
The second step is to apply Theorem~\ref{thm:distance-to-desert} to obtain a "desert automaton" $\aut_d$ recognising the cost function $\sem{\aut'}$.
The third step is to construct a "desert automaton" $\aut$ guessing the decomposition:
\[
\sem{\aut}(w) =\max_{\substack{w^{\Jc}, b^{\Jc} \text{ decompositions} \\ \text{of } w, b}} \sem{\aut_d}(w^{\Jc} b^{\Jc}).
\]
By construction, $\aut$ recognises the cost function $\Phi$.
\end{proof}

Theorem~\ref{thm:main-sfp} follows by combining Theorem~\ref{thm:boundedness-of-desert} and Lemma~\ref{lem:desert}.

%% file: complexity.tex
We now provide a PSPACE lower bound for "simple SFP", by reduction from "NFA intersection emptiness"~\cite{Kozen77}.

\begin{problem}[""NFA intersection emptiness""]
\AP
Given $k$ non-deterministic finite automata $\aut_1, \dots \aut_k$, decide whether $\bigcap_i \L(\aut_i) = \varnothing$.
\end{problem}

Recall that the instance discussed in Figure~\ref{fig:flow1} is an instance of "simple SFP". 
It can be used as a gadget to encode intersection of two NFAs, essentially by encoding in action $c$ a synchronous run in both NFA, one from $q_2$ to $q_1$ and one from $q_3$ to $q_3$. 

\paragraph*{The ""gadget"" $\gadg k$}
\AP
We now introduce and study a "gadget" $\gadg k$, which generalizes the one from Figure~\ref{fig:flow1} and is tailored to encode synchronous runs of $k$ NFA. It is illustrated in Figure~\ref{fig:flow3}.

\begin{figure}[!ht]
    \centering
    \def\svgwidth{\columnwidth}
    \resizebox{0.6\textwidth}{!}{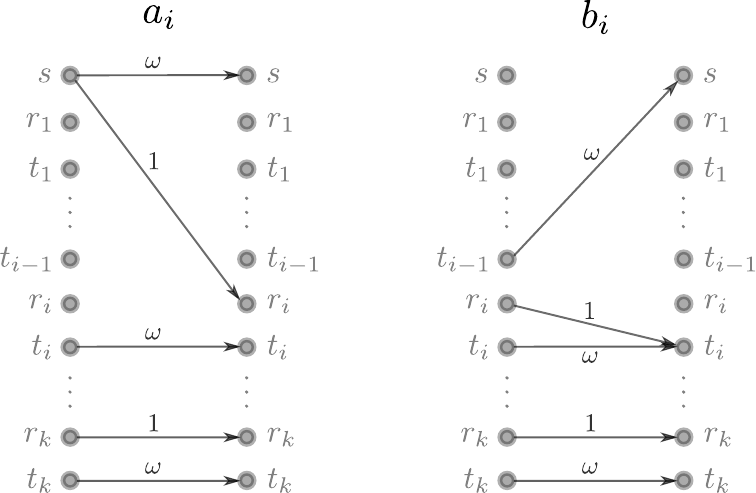}
    \caption{The capacities $a_i$ and $b_i$ which define the gadget $\gadg k$}
    \label{fig:flow3}
\end{figure}

\begin{itemize}
\item There are $2k+1$ "states", $s, r_1,t_1,r_2,t_2, \dots, r_k,t_k$. The "source" is $s$ and the "target" is $t_k$. Intuitively, to move a "token" from $s$ to $t_i$, one must first move it into $r_i$, then recursively move all "tokens" to $t_{i-1}$, before the "token" in $r_i$ is allowed to move to $t_i$, which drives all other "tokens" back to the "source" $s$.
\item There are $2k$ "capacities" $a_1,b_1,\dots, a_k,b_k$. The role of $a_i$ is to allow to move a "token" to $r_i$, provided all "tokens" are in $\{s, t_i, r_{i+1}, t_{i+1}, \dots, r_k, t_k\}$. The role of $b_i$ is to allow a "token" to move from $r_i$ to $t_i$, provided all "tokens" are in\footnote{For this definition when $i=1$, we identify $s$ and $t_{i-1}$.} $\{t_{i-1}, r_i, t_i, \dots, r_k, t_k\}$. Additionally, $b_i$ resets all "tokens" from $t_{i-1}$ back to $s$. Formally, we have
\[
a_i(q,q')=\begin{cases}
\omega \text{ if } q=q' \in \{s, t_i, t_{i+1}, \dots, t_k\} \\
1 \text{ if } q=q' \in \{\r_{i+1}, \dots, r_k \} \text{ or } (q,q')=(s,r_i) \\
0 \text{ otherwise},
\end{cases}
\]
and
\[
b_i(q,q')=\begin{cases}
\omega \text{ if } q=q' \in \{t_i, t_{i+1}, \dots, t_k\} \text{ or } (q,q') = (t_{i-1},s) \\
1 \text{ if } q=q' \in \{\r_{i+1}, \dots, r_k \} \text{ or } (q,q')=(r_i,t_i) \\
0 \text{ otherwise}.
\end{cases}
\]

\end{itemize}

Towards introducing our reduction from "NFA intersection", we study the behaviour of $\gadg k$. Let us first see that it is a positive instance of "simple SFP".

\begin{lemma}\label{lem:existsword}
For all $n$, there exists a "word of capacities" $w^{(n)} \in \Capa^*$ and a "flow word" $\bfl f^{(n)} \leq w^{(n)}$ such that $n \cdot s \flowpath{\bfl f^{(n)}} n \cdot t_k$.
\end{lemma}

\begin{proof}
We prove by induction over $i \in \{1, \dots, k\}$ that from any "configuration" $c$ such that $c(s)=n, c(r_1)=c(t_1)= \dots = c(r_{i-1})=c(t_{i-1})=c(r_i)=0$, $c(r_{i+1})=c(r_{i+2}) = \dots = c(r_k)=1$, and $c(t_i), c(t_{i+1}) \dots c(t_k)$ are arbitrary, there exists a "capacity word" $w \in \{a_1,b_1, \dots, a_i,b_i\}^*$ and a "flow word" $\bfl f \leq w$ such that $c \flowpath{\bfl f} c'$, where $c'$ is identical to $c$ except for a "token" which has moved from $s$ to $t_i$. For $i=0$, the result is easily obtained with $w=a_1 b_1$. For the induction step, note that "capacities" from $\{a_1, b_1, \dots, a_{i-1}, b_{i-1}\}$ allow for "flows" which leave the positions $t_i,r_{i+1},t_{i+1},r_{i+2}, \dots, r_k, t_k$ unchanged (with value 1 at states $r_j$ and arbitrary at states $t_j$). We first use "capacity" $a_i$ to move one "token" from $s$ to $r_i$, then apply the induction hypothesis $n-1$ times to move all remaining "tokens" to $t_{i-1}$, and finally use "capacity" $b_k$ to reach the wanted "configuration" $c'$. This concludes the proof of the induction.

To prove the lemma, it suffices to apply this inductive statement $n$ times for $i=k$, to move each "token" from $s$ to $t$. Unravelling the inductive structure of the proof yields that the "word of capacities" $a_k(a_{k-1}(\dots(a_1 b_1)^{n-k+1}\dots)^{n-2} b_{k-1})^{n-1} b_k $ "goes from" $n \cdot s$ to $(n-1) \cdot s + t_k$.
\end{proof}

Note that the "capacity word" from the proof above has $a_k a_{k-1} \dots a_1$ as a prefix: the "strategy" we use starts by moving one "token" in each $r_i$, and then applies "capacity" $b_i$ which is able to \emph{synchronously} move each of these single "tokens" from $r_i$ back to $r_i$. This behaviour in $\gadg k$ is crucial in our reduction, and formalized now.

\begin{lemma}\label{lem:synchronize}
Let $n \geq k$, and $w \in \{a_1,b_1, \dots, a_k, b_k\}^*$ be a "word of capacities" of size $N$, and $\bfl f = f_1 \dots f_N$ a "flow word" $\leq w$ such that $n \cdot s \flowpath{\bfl f} n t_k$. Let $c_l=\pre(f_l)$ for $l \in \{1, \dots, N\}$. Then there exists an index $l$ such that for all $i \in \{1, \dots, k\}$, $c_l(r_i)=1$.
\end{lemma}

\begin{proof}
We prove by induction over $i \in \{1, \dots, k\}$ that for any $n \geq i$, a "flow word" $\bfl f = f_1 \dots f_N$ smaller than some $w \in \{a_1, b_1, \dots, a_i, b_i\}^*$ "going from" $n \cdot s + r_{i+1} + \dots + r_k$ to $n t_i + r_{i+1} + \dots + r_k$ must have some $c_l = \pre(f_l)$ satisfying for all $j$, $c_l(r_j) = 1$.
For $i=1$, this is true for $c_1=n \cdot s + r_1 + \dots + r_k$. Let $i>0$ and assume the result proved for smaller $i$'s. For all $l \in \{1, 2, \dots, N-1\}$, let $c_l = \pre(f_l)$ (which is also equal to $\post(f_{l-1})$ for $l \geq 2$). Note that since $f_l \leq w_l \in \{a_1, b_1, \dots, a_i, b_i\}$ by hypothesis, for all $j \geq i+1$, since $c_1(r_j)=1$ and since the $r_j$ is the unique "state" $q$ such that $w_l(q,r_j)=1$, an easy induction concludes that $c_l(r_j)=1$ and $f_l(r_j,r_j)=1$. Let $l_{\righ} = \min_l(c_l(t_i)>0)$. Since $c_{l_{\righ}-1}(t_i)=0$, and $w_l \in \{a_1,b_1, \dots, a_i,b_i\}$, it must be that $w_{l_\righ -1} = b_i$ and hence $c_{l_\righ -1} = (n-1) \cdot t_{i-1} + r_i + r_{i+1} + \dots + r_k$. Let $l_{\lef} = \min_l(c_l(r_i)>0) \in \{2, 3, \dots, l_{\righ -1}\}$.
Since $c_{l_\lef -1}(r_1)=0$ and $w_{l_\lef -1} \in \{a_1,b_1, \dots, a_i,b_i\}$, it must be that $w_{l_\lef -1} = a_i$ and hence $c_{l_\lef} = n \cdot s + r_i + r_{i+1} + r_{i+2} + \dots + r_k$. Now, an easy induction concludes that for all $l \in \{l_\lef, l_\lef +1, \dots, l_\righ -1\}$, we have $c_l(r_i)=1$, so it must be that $w_l \in \{a_1, b_1, \dots, a_{i-1}, b_{i-1}\}$. We conclude by invoking the induction hypothesis between $l_\lef$ and $l_\righ -1$.
\end{proof}

We now have a sufficiently good understanding of the "gadget" $\gadg k$ to present our lower bound.

We fix $k$ finite non-deterministic automata $\aut_1=(\states_1, I_1,F_1,\Delta_1), \dots \aut_k=(\states_k, I_k, F_k, \Delta_k)$ on a common alphabet $\Sigma$. Without loss of generality $\states_1, \dots, \states_k$ are disjoint sets and $\Delta_i \subseteq \states_i \times \Sigma \times \states_i$. To encode intersection of $\aut_1, \dots, \aut_k$ in "simple SFP", we add to $\gadg k$ a copy of each automaton, an additionnal "state" $r'_1$, "capacities" $c_\sigma$ for each $\sigma \in \Sigma$ and two "capacities" $\sr$ and $\fr$. The idea is as follows: to proceed from $r_1 + r_2 + \dots + r_k$ towards $t_1 + r_2 + \dots + r_k$, one now has to first go to $r_1' + r_2 + \dots + r_k$, which requires each of the $k$ tokens in $r_1, \dots, r_k$ to synchronously follow a run in $\aut_1, \dots, \aut_k$ respectively. Formally, we let $\states=\{s, r_1, t_1, \dots, r_k, t_k, r_1'\} \cup Q_1 \cup \dots \cup Q_k$, $\Capa = \{a_1, b_1, \dots, a_k, b_k, \sr, \fr\} \cup \{c_\sigma, \sigma \in \Sigma\}$. We let $a_i$ and $b_i$, for $i \in \{1, \dots, n\}$ be defined just as previously, except for $b_1$ which is now defined as
\[
b_1(q,q')=\begin{cases}
\omega \text{ if } q=q' \in \{s, t_1, \dots, t_k\}, \\
1 \text{ if } q=q' \in \{r_2, \dots, r_k \} \text{ or } (q,q')=(r'_1,t_1), \\
0 \text{ otherwise}.
\end{cases}
\]
Moreover, the "capacity" $c_\sigma$ is given by
\[
c_\sigma(q,q')=\begin{cases}
\omega \text{ if } q=q' \in \{s, t_1, \dots, t_k\}, \\
1 \text{ if for some } i \in \{1, \dots, k\}, (q, \sigma, q') \in \Delta_i, \\
0 \text{ otherwise},
\end{cases}
\]
the "capacity" $\sr$ is given by 
\[
\sr(q,q')=\begin{cases}
\omega \text{ if } q=q' \in \{s, t_1, \dots, t_k\}, \\
1 \text{ if for some } i \in \{1, \dots, k\}, q=r_i \text{ and } q' \in I_i, \\
0 \text{ otherwise},
\end{cases}
\]
and the "capacity" $\fr$ is given by
\[
\fr(q,q')=\begin{cases}
\omega \text{ if } q=q' \in \{s, t_1, \dots, t_k\}, \\
1 \text{ if } q \in F_1 \text{ and } q'=r'_1, \\
1 \text{ if for some } i \in \{2, \dots, k\}, q \in F_i \text{ and } q' = r_i, \\
0 \text{ otherwise}.
\end{cases}
\]

We now prove that this construction indeed gives a reduction from "NFA intersection" to "simple SFP".

\begin{theorem}
The intersection of $\L(\aut_i)$ is non-empty if and only if for all $n$, there is a "capacity word" $w \in \Capa^*$ and a "flow word" $\bfl f \leq w$ such that $n \cdot s \flowpath{\bfl f} n \cdot t_k$.
\end{theorem}

\begin{proof}
Let us first assume that $\bigcap_i \L(\aut_i) \neq \varnothing$, and let $u \in \bigcap_i \L(\aut_i)$. Then it is straightforward to adapt the proof of Lemma~\ref{lem:existsword} to show that the "word of capacities" 
\[
w_n = a_k(a_{k-1}(\dots(a_1 \cdot \sr \cdot u \cdot \fr \cdot b_1)^{n-k+1}\dots)^{n-2} b_{k-1})^{n-1} b_k
\]
goes from $n \cdot s$ to $(n-1) \cdot s + t_k$.
Consequently, the "word of capacities"\footnote{With the convention that $v^n$ is the empty word when $n$ is a negative integer} 
$w = w_n w_{n-1} \dots w_1$ "goes from" $n \cdot s$ to $n \cdot t_k$.

Conversely, assume that for all $n$ there is a "word of capacities" $w \in \Capa^*$ and a "flow word" $\bfl f \leq w$ "going from" $n \cdot s$ to $n \cdot t_k$, and let $n = k$. 
Let $N = |w|$ and let $c_l = \pre(f_l)$ for $l \in \{1, \dots, N\}$. 
Again, by a straightforward adaptation of Lemma~\ref{lem:synchronize} there exists an index $l_0$ such that for all $i \in \{1, \dots, k\}$, $c_{l_0}(r_i)=1$. This implies that $w_{l_0+1} = \sr$, since for any other "capacity" $c$ we have $\pre(c) \ngeq c_{l_0}$. Hence, there are initial states $q^{(1)} \in I_1, q^{(k)} \in I_k$ in each automaton such that $c_{l_0+1} = q^{(1)} + \dots + q^{(k)} + p$, with $p$ a "configuration" of support $\subseteq \{s,t_1,t_2, \dots, t_k\}$.
Let $l_1 = \min\{l \geq l_0+1 \mid w_l \notin \{c_\sigma, \sigma \in \Sigma\}\}.$ Let $L=l_1 - l_0 -1 $, and $u=u_1 \dots u_{L}$ be such that for all $l \in \{1, \dots, L\}$, $w_{l_0+l} = c_{u_i}$. An easy induction shows that for all $l \in \{1, \dots, L\},$ there are "states" $q^{(1)}_l \in \states_1, \dots, q^{(k)}_l \in \states_l$ such that $c_{l_0+l} = q^{(1)}_l + \dots + q^{(k)}_l + p$, and that for each $i \in \{1, \dots, k\}$, $q^{(i)}_1 \xrightarrow{u_1} q^{(i)}_2 \xrightarrow{u_2} \dots \xrightarrow{u_L} q^{(i)}_l$ is a run in $\aut_i$.
Since it is the only "capacity", other than the $c_\sigma$'s, with arrows leaving from $\bigcup \states_i$, it must be that $w_{l_0+L} = \fr$. This implies that for all $i \in \{1, \dots, k\}, q^{(i)}_L \in F_i.$ We conclude that $u$ belongs to each of the $\L(\aut_i)$.
\end{proof}

%% file: flows3.pdf_tex
\begingroup%
  \makeatletter%
  \providecommand\color[2][]{%
    \errmessage{(Inkscape) Color is used for the text in Inkscape, but the package 'color.sty' is not loaded}%
    \renewcommand\color[2][]{}%
  }%
  \providecommand\transparent[1]{%
    \errmessage{(Inkscape) Transparency is used (non-zero) for the text in Inkscape, but the package 'transparent.sty' is not loaded}%
    \renewcommand\transparent[1]{}%
  }%
  \providecommand\rotatebox[2]{#2}%
  \ifx\svgwidth\undefined%
    \setlength{\unitlength}{362.42639368bp}%
    \ifx\svgscale\undefined%
      \relax%
    \else%
      \setlength{\unitlength}{\unitlength * \real{\svgscale}}%
    \fi%
  \else%
    \setlength{\unitlength}{\svgwidth}%
  \fi%
  \global\let\svgwidth\undefined%
  \global\let\svgscale\undefined%
  \makeatother%
  \begin{picture}(1,0.6528291)%
    \put(0,0){\includegraphics[width=\unitlength,page=1]{flows3.pdf}}%
  \end{picture}%
\endgroup%

%% file: conclusions.tex
We showed the decidability of the "stochastic control problem".
Our approach uses "well quasi orders" and the "sequential flow problem", which is then solved using the theory of "regular cost functions". 
As an intermediate we also introduce the "simple sequential flow problem" as a problem of independent interest, 
and provide an \textsc{EXPSPACE} complexity upper bound as well as a \textsc{PSPACE} complexity lower bound.

Together with the original result of~\cite{BDGG17,BDGGG19} in the adversarial setting, our result contributes 
to the theoretical foundations of "parameterised control".
We return to the first application of this model, control of biological systems.
As we discussed the stochastic setting is perhaps more satisfactory than the adversarial one, 
although in our examples complicated behaviours also exist in the stochastic setting,
which are arguably not relevant for modelling biological systems.

We thus pose two open questions.
The first is to settle the complexity status of the "stochastic control problem".
Mascle, Shirmohammadi, and Totzke~\cite{MST19} proved the \textsc{EXPTIME}-hardness of the problem, 
which is interesting because the underlying phenomena involved in this hardness result 
are specific to the stochastic setting (and do not apply to the adversarial setting).
Our algorithm does not even yield elementary upper bounds, leaving a very large complexity gap.
We believe that a good understanding of the "sequential flow problem" and the structure of its positive instances can be useful to make progress in this direction. Closing the complexity gap by either finding a better complexity lower bound or another algorithmic technique appears to be an interesting and challenging problem.

The second question is towards more accurately modelling biological systems:
can we refine the stochastic control problem by taking into account the synchronising time of the controller,
and restrict it to reasonable bounds?